\documentclass[runningheads]{llncs}
\usepackage[T1]{fontenc}
\usepackage{graphicx}
\usepackage{bm}
\usepackage{amsmath}
\usepackage{float}
\usepackage{diagbox}
\usepackage{amsfonts}
\begin{document}
\title{Scheduling two types of jobs with minimum makespan}
%
%
\author{Song Cao\orcidID{0009-0002-1760-3820} \and
Kai Jin\orcidID{0000-0003-3720-5117}}
\authorrunning{Cao and Jin}
%
\institute{Shenzhen Campus of Sun Yat-sen University, Shenzhen, Guangdong, China
\email{caos6@mail2.sysu.edu.cn}\\
\email{jink8@mail.sysu.edu.cn}\\
}
\maketitle              
\begin{abstract}
We consider scheduling two types of jobs (A-job and B-job) to $p$ machines and minimizing their makespan. A group of same type of jobs processed consecutively by a machine is called a batch. For machine $v$, processing $x$ A-jobs in a batch takes $k^A_vx^2$ time units for a given speed $k^A_v$, and processing $x$ B-jobs in a batch takes $k^B_vx^2$ time units for a given speed $k^B_v$. 
We give an $O(n^2p\log(n))$ algorithm based on dynamic programming and binary search for solving this problem,
 where $n$ denotes the maximal number of A-jobs and B-jobs to be distributed to the machines.
Our algorithm also fits the easier linear case where each batch of length $x$ of $A$-jobs takes $k^A_v x$ time units and each batch of length $x$ of $B$-jobs takes $k^B_vx$ time units. The running time is the same as the above case.

\end{abstract}
\section{Introduction}

We consider scheduling two types of jobs A and B to machines and minimizing their makespan, i.e., the maximum time completing all jobs. 
Suppose $n_a$ A-jobs and $n_b$ B-jobs have to be scheduled, and can be distributed to $p$ machines. 
A group of A-jobs processed consecutively by some machine is called an \emph{A-batch},
  and a group of B-jobs processed consecutively by some machine is called a \emph{B-batch}.
    A batch refers to an A-batch or a B-batch. A machine processing a batch continuously will overheat over time, resulting in decreased performance. This leads to a quadratic relationship between processing time and the number of jobs in a batch.
For the $v$-th machine, processing $x$ A-jobs in a batch takes $k^A_vx^2$ time units, 
         and processing $x$ B-jobs in a batch takes $k^B_vx^2$ time units, 
            where $k^A_v,k^B_v$ are given speed. 
Moreover, for the $v$-th machine,
   assume that it takes $t^A_v$ extra time units to process any A-batch (the overhead to switch the status to processing A-jobs), 
  and $t^B_v$ extra time units to process any B-batch (the overhead to switch the status to processing B-jobs).
No empty batch is allowed, and two neighbouring batches of any machine must be of different types.
 Denote $n=\max(n_a,n_b)$.

This paper gives an $O(n^2p\log(n))$ time algorithm for solving the above problem based on binary search.
Given a parameter $L$, we have to determine whether there is a scheduling with makespan not exceeding $L$ time units.
  
For convenience, we use a pair $(a,b)$ to indicate the task-combination formed by $a$ A-jobs and $b$ B-jobs
  (it is always assumed that $a\geq 0$ and $b\geq 0$).  
Consider the set $B$ of $b$, for which the machines can finish the task-combination $(n_a,b)$ within $L$ time units. 
We determine that the makespan can be smaller or equal to $L$ time units if and only if $n_b\in B$.
On the other hand, we prove that $B$ is in fact an interval, which can be computed by dynamic programming that takes $O(n^2p)$ time.
This leads to our final algorithm that runs in $O(n^2p\log(n))$ time. 
The major challenge for designing the algorithm lies in proving the theorem that says $B$ is an interval, 
   for which a lot of analysis are shown in this paper.

We also show that our algorithm can solve the linear case where each $A$-batch of length $x$ takes $k^A_v x$ time units and each $B$-batch of length $x$ takes $k^B_vx$ time units. The running time is the same as the quadratic case mentioned above.

\subsection{Related works}

Machine scheduling problem is widely studied for its practical importance. Consider scheduling jobs $J_1, \ldots, J_n$ on parallel machines $M_1, \ldots, M_m$ and minimizing the makespan. Preemption(job splitting), precedence relation and release date are not considered for each job. The processing time that machine $M_i$ processes job $J_j$ is denoted as $p_{ij}$. Classified by machine environment, when $p_{ij}=p_j$, it's called identical parallel machine; when $p_{ij}=p_j/s_i$ for a given speed $s_i$ of machine $M_i$, it's called uniform parallel machine; and the general case, when $p_{ij}=p_j/s_{ij}$, it's called unrelated parallel machine. According to the three-field classification note \cite{Lawler}, the above three cases can be denoted as $P||C_{max}$, $Q||C_{max}$ and $R||C_{max}$. The problem we solve in this paper is a variation of $R||C_{max}$.

$R||C_{max}$ is NP-hard since its simple version P||Cmax has been proved NP-hard \cite{Garey}. Studies on these problems focus on finding polynomial $\rho$-approximation algorithm ($\rho>1$), that is, polynomial time algorithm which generates solution at most $\rho$ times the optimal solution\cite{Lawler}. Lenstra et al.\cite{Lenstra} proved that for $R||C_{max}$, the worst-case ratio for approximation is at least $3/2$ unless $P=NP$. And they present a polynomial $2$-approximation algorithm. Ghirardi et al.\cite{Ghirardi} provides an heuristic called Recovering Beam Search, which generates approximate solutions in $O(n^3m)$ time.

Different variations of unrelated parallel machine problem have been studied. Consider minimizing $L_p$ norm of unrelated machines' completion time, Alon et al.\cite{Alon} provides an $(1+\varepsilon)$-approximation scheme, that is, a family of algorithms $\{A_{\varepsilon}\}$ ($\varepsilon>0$) such that, each $\{A_{\varepsilon}\}$ is $1+\varepsilon$-approximation, and its running time may depend on input and $\varepsilon$\cite{Lawler}. Azar et al. \cite{Azar} provides a $2$-approximation algorithm for $p>1$, and $\sqrt{2}$-approximation for $p=2$. And Kumar et al.\cite{Kumar} improves it with a better-than-$2$ approximation algorithm for all $p\geq 1$. Im and Li\cite{Im} consider weighted completion time with job release date. They improve a $2$-approximation to $1.8786$-approximation for both preemptive and non-preemptive problems. For $R||C_{max}$, when matrix rank formed by job processing time is small, Bhaskara et al.\cite{Bhaskara} show that it admits an QPTAS for rank $2$, and is APX-hard for rank $4$. Chen et al.\cite{Chen} continue this research and prove it APX-hard for rank $3$. Deng et al.\cite{Deng} consider a more general problem called GLB problem, where the load of each machine is a norm, and minimizes generalized makespan(also a norm). They provide a polynomial $O(\log n)$-approximation algorithm. Im and Li\cite{Im} give a $1.45$-approximation algorithm for minimizing total weighted completion time, and a $\sqrt{4/3}$-approximation for minimizing $L_2$-norm. Bamas et al.\cite{Bamas} obtain a polynomial $(2-1/n^{\varepsilon})$-approximation for two-value makespan minimization, and a polynomial $(1.75+\varepsilon)$-approximation for the restricted assignment case.

The organization of this paper is as follows. Section \ref{sect:single-day} considers the one machine case and proves several important observations for this case.
 Section \ref{sect:multiple-day} considers the multiple machines case and gives our main algorithm.

\section{The one machine case}\label{sect:single-day}

Consider any fixed machine $v$ with parameters $k^A_v,k^B_v,t^A_v$ and $t^B_v$.
Denote by $f(a,b)$ the minimum time units for this machine to finish task-combination $(a,b)$,
 and denote by $S_v$ the set of task-combinations that can be finished by machine $v$ within $L$ time units 
   (where $L$ is a bound that will be specified by a binary search in the main algorithm).
  Formally, $S_v=\{(a,b) \mid f(a,b) \leq L\}$.
  
A crucial property of $S_v$ is described in the following theorem.


\begin{theorem}\label{thm:S}    
    For fixed $a\geq 0$, those $b$ for which $(a,b)\in S_v$ are consecutive.
\end{theorem}

To prove this theorem, we need some knowledge about $f(a,b)$ as shown below.

\begin{definition}
For $a\geq 0$ and $s\geq 0$, define $cost^A(a,s)$ to be
  the minimum time units for machine $v$ to process $a$ A-jobs in $s$ batches
  (here, we allow empty batch length and hence $cost^A(a,s)$ is well-defined even for the case $a<s$).
  Formally, $$cost^A(a,s) = min_{\begin{subarray}{c}x_1+\ldots+x_s=a \\ x_1\geq 0,\ldots,x_s\geq 0\end{subarray}}
    \{s\cdot t^A_v + (x_1^2+\ldots +x_s^2)\cdot k^A_v\}.$$

  Note that for $s=0$, we have $cost^A(0,0)=0$ and $cost^A(a,0)=\infty$ for $a>0$.
      
\smallskip Similarly, define
    $$cost^B(b,s) = min_{\begin{subarray}{c}x_1+\ldots+x_s=b \\ x_1\geq 0,\ldots,x_s\geq 0\end{subarray}}
         \{s\cdot t^B_v + (x_1^2+\ldots +x_s^2)\cdot k^B_v\}.$$

\end{definition}

For $a\geq 0$ and $s>0$, the following equation holds (trivial proof omitted).
\begin{equation}\label{eqn:trivial}
  cost^A(a,s)=s\cdot t^A_v+\left[(a\bmod s)(\left[\frac{a}{s}\right]+1)^2+(s-(a\bmod s))\left[\frac{a}{s}\right]^2\right]\cdot k^A_v.
\end{equation}

According to (\ref{eqn:trivial}), for fixed $s\geq 1$, 
function $cost^A(a,s)$ of variable $a$ increases on $[0,\infty)$. Similarly, function $cost^B(b,s)$ of variable $b$ increases on $[0,\infty)$.

\begin{lemma}\label{lemma:f-mcost}
 $f(a,b)=\mathop{\min}_{0\leq s\leq \min(b,a+1)}cost(a,b,s)$, where 
 \begin{equation}\label{eqn:def_costabs}
   cost(a,b,s)=cost^B(b,s)+mcost^A(a,s),
 \end{equation}
 and
 \begin{equation}
    mcost^A(a,s)=\min_{s'\in \{s-1,s,s+1\},s'\geq 0}cost^A(a,s').
 \end{equation}
\end{lemma}

Lemma~\ref{lemma:f-mcost} is not that obvious --
An empty batch is allowed in the definition of $cost^A$ and $cost^B$,
  yet \textbf{not} allowed in processing the task-combination $(a,b)$.

\begin{proof}
  Suppose we want to process the task combination $(a,b)$. Without loss of generality,
    assume that we use $s$ amount of $B$-batches. Obviously,
  (1) $s\leq b$ (otherwise some $B$-batches would be empty, which is not allowed).
  (2) $s\leq a+1$ (since the number of $A$-batches is at least $s-1$ and at most $a$).
  Together, $$0\leq s\leq \min(b,a+1).$$
  
  Under the assumption that there are $s$ amount of $B$-batches and $s\leq b$,
    we know that all the $B$-jobs can be done in $cost^B(b,s)$ time units.
  It remains to show that the minimum time units for completing $A$-jobs is $mcost^A(a,s)$ time
    under the constraint that the number of $A$-batches belongs to $\{s-1,s,s+1\} \cap [0,a]$.

  Plugging in the formula of $mcost^A(a,s)$, it reduces to proving that 
     \begin{equation}\label{eqn:proof-mcost}
        \min_{s'\in \{s-1,s,s+1\}\cap [0,+\infty)}cost^A(a,s') = \min_{s'\in \{s-1,s,s+1\}\cap [0,a]}cost^A(a,s').
     \end{equation}
  
  For $s<a$, (\ref{eqn:proof-mcost}) holds as $\{s-1,s,s+1\}\cap [0,+\infty)=\{s-1,s,s+1\}\cap [0,a]$.
  
  For $s=a$, observing that $cost^A(a,a+1)>cost^A(a,a)$ (due to the definition of $cost^A$),
    the left and right sides of (\ref{eqn:proof-mcost}) both equal $\min_{s'\in \{a-1,a\}\cap [0,a]}cost^A(a,s')$.
    
  For $s=a+1$, observing that $cost^A(a,a+2)>cost^A(a,a)$ and $cost^A(a,a+1)>cost^A(a,a)$,
    the two sides of (\ref{eqn:proof-mcost}) are equal to $\min_{s'\in \{a\}}cost^A(a,s')$.  
\qed
\end{proof}

\begin{lemma}\label{lemma:cost^A-diff}
For $a \geq 0$, all the functions below are increasing functions of $s$:
\begin{enumerate}
\item $cost^A(a,s+1)-cost^A(a,s)$ (namely, the difference function of $cost^A(a,s)$);
\item $cost^B(a,s+1)-cost^B(a,s)$ (namely, the difference function of $cost^B(a,s)$);
\item $mcost^A(a,s+1)-mcost^A(a,s)$ (namely, the difference function of $mcost^A(a,s)$).
\end{enumerate}
In other words, $cost^A(a,s),cost^B(a,s),mcost^A(a,s)$ are convex functions of $s$.
\end{lemma}

Due to space limit, the proof of Lemma~\ref{lemma:cost^A-diff} is deferred to appendix.

\begin{lemma}\label{cost-lemma}
 \begin{enumerate}
   \item For fixed $a \geq 0$ and $b \geq 0$, function $cost(a,b,s+1)-cost(a,b,s)$ of $s$ (namely, the difference function of $cost(a,b,s)$)
       is increasing.
   \item For fixed $s\geq 0$ and $a \geq 0$, function $cost(a,b,s)$ of $b$ is increasing.
   \item For fixed $s\geq 0$ and $b \geq 0$, function $cost(a,b,s)$ of $a$ is increasing.
   \item Function $cost(x,x+1,x+1)$ of $x$ is increasing.
 \end{enumerate}
\end{lemma}

\begin{proof}
1. Recall $cost(a,b,s)=cost^B(b,s)+mcost^A(a,s)$. 
 The claim on $cost(a,b,s)$ follows from the claims on $cost^B$ and $mcost^A$ given in lemma \ref{lemma:cost^A-diff}.

2. This claim follows from the fact that the function $cost^B(b,s)$ of variable $b$ increases on $[0,\infty)$ (which is stated right below (\ref{eqn:trivial})).

3. We need to distinguish two cases by whether $s=0$ or $s>0$. First, suppose $s>0$. It reduces to showing that $mcost^A(a,s)$ is increasing as $a$ grows.
    Recall $mcost^A(a,s)= \min\left(cost^A(a,s-1),cost^A(a,s),cost^A(a,s+1)\right)$.
    This function is increasing 
      because of the fact that $cost^A(a,s-1),cost^A(a,s),cost^A(a,s+1)$ are all increasing functions of $a$ (as stated right below (\ref{eqn:trivial})).
      The proof for the case $s=0$ is similar and easier; we leave it as an exercise for the reader.

4.  Notice that
    $cost(x,x+1,x+1) = cost^B(x+1,x+1)+mcost^A(x,x+1) 
         = cost^B(x+1,x+1)+cost^A(x,x) = (x+1)(k^B_v+t^B_v)+x(k^A_v+t^A_v).$
 \qed
\end{proof}

We are ready to prove Theorem~\ref{thm:S}. Suppose $a\geq 0$ is fixed in the following. 

We introduce a table $M$ with $n_b+1$ rows and $a+2$ columns to prove Theorem~\ref{thm:S}:
for $(b,s)$ where $0\leq b\leq n_b$, and $0\leq s\leq a+1$, define
\begin{equation}
  M(b,s)=
  \left\{
  \begin{array}{cc}
    [cost(a,b,s)\leq L],  & s\leq b;\\
    -~(\text{ leave it as undefined}), & s>b,
  \end{array}\right.
\end{equation}
where $[\cdot]$ denotes the Iverson bracket. See table \ref{table1-example} for an example.
Rows in $M$ are indexed with $0$ to $n_b$. Columns in $M$ are indexed with $0$ to $a+1$.

\begin{table}
  \centering
  \renewcommand\arraystretch{1.5}
  \setlength{\tabcolsep}{5mm} {
  \begin{tabular}{|c|c|c|c|c|}
\hline
\diagbox[width=0.12\textwidth,height=9mm]{$b$}{$s$} & 0 & 1 & 2 & 3 \\
\hline
0 & \textbf{1} & - & - & - \\
\hline
1 & \textbf{1} & \textbf{1} & - & - \\
\hline
2 & 0 & 0 & \textbf{1} & - \\
\hline
3 & 0 & 0 & \textbf{1} & 0 \\
\hline
4 & 0 & 0 & 0 & 0 \\
\hline
\end{tabular}                 }
  \caption{An example of $M$, where $n_b=4$ and $a=2$.}\label{table1-example}
\end{table}

\begin{proof}[of Theorem~\ref{thm:S}]
Assume $a\geq 0$ is fixed.
Recall that this theorem states that $\{b \mid f(a,b)\leq L\}$ are consecutive.
Following Lemma~\ref{lemma:f-mcost}, $f(a,b)\leq L$ if and only if the $b$-th row in $M$ has a $1$.
Hence it reduces to showing that (i) the rows in $M$ with at least one $1$ are consecutive.

We now point out two facts about table $M(b,s)$.
\begin{enumerate}
  \item[Fact~1.] If $M(b,s)=1$ and $s\leq b-1$, then $M(b-1,s)=1$.
  \begin{proof}
    Since $M(b,s)=1$, we know $cost(a,b,s)\leq L$. By lemma \ref{cost-lemma}.2 we know $cost(a,b-1,s)\leq cost(a,b,s)\leq L$. So $M(b-1,s)=1$. \qed
  \end{proof}
  \item[Fact~2.] If $M(s,s)=1$ and $M(s+1,s+1)=0$, then $M(s',s')=0$ for $s'>s+1$.
  \begin{proof}
    Since $M(s,s)=1$ and $M(s+1,s+1)=0$, we know $cost(a,s,s)\leq L<cost(a,s+1,s+1)$. So,
    \begin{equation}\label{cost(a,s,s)<cost(a,s+1,s+1)}
      cost(a,s,s)<cost(a,s+1,s+1).
    \end{equation}
    
    Observe that $cost(a,s',s')=cost^B(s',s')+mcost^A(a,s')=s't_v^B+s'k^B_v+mcost^A(a,s')$. 
        Following lemma \ref{lemma:cost^A-diff}, function $cost(a,s',s')$ of $s'$ is convex. 
    Further since (\ref{cost(a,s,s)<cost(a,s+1,s+1)}), $cost(a,s',s')$ is strictly increasing when $s'\geq s+1$. 
    Therefore $cost(a,s',s')>cost(a,s+1,s+1)>L$, which means $M(s',s')=0$.  \qed
  \end{proof}
\end{enumerate}

According to Fact~2, the 1's in the diagonal of $M$ are consecutive.
According to Fact~1, the 1's in any particular column are consecutive and the lowest 1 (if any) in that column always appears at the diagonal of $M$.
Together, we obtain argument (i). \qed
\end{proof}

\subsection{More properties of $S_v$}


For $0\leq a\leq n_a,0\leq b\leq n_b$, define $F^{(v)}(a,b)=[(a,b)\in S_v]=[f(a,b)\leq L]$.
Abbreviate $F^{(v)}$ as $F$ when $v$ is clear.  
This table represents $S_v$. See table \ref{table2-example} for an illustration.
Whereas Theorem~\ref{thm:S} shows the consecutiveness of 1's within any row of $F$,
  the next theorem shows properties of $F$ between its consecutive rows.

\begin{theorem}\label{thm:S-more}
For any fixed $a\geq 0$,
\begin{enumerate}
    \item If some row of $F$ are all 0's, so is the next row of $F$.
    \item If both row $a$ and row $a+1$ contain 1, at least one of the following holds:\\
      (1) there is $b$ such that $F(a,b)=F(a+1,b)=1$.\\
      (2) $F(a,a+1)=F(a+1,a+2)=1$.
  \end{enumerate}
\end{theorem}

\begin{table}
  \centering
  \renewcommand{\arraystretch}{1.5}
  \setlength{\tabcolsep}{5mm}
  \begin{tabular}{|c|c|c|c|c|c|c|c|c|c|}
\hline
\diagbox[width=0.12\textwidth,height=9mm]{$a$}{$b$} & $0$ & $1$ & $2$ & $3$ \\
\hline
$0$ & \textbf{1} & \textbf{1} & \textbf{1} & \textbf{1} \\
\hline
$1$ & 0 & \textbf{1} & 0 & 0 \\
\hline
$2$ & 0 & 0 & 0 & 0 \\
\hline
$3$ & 0 & 0 & 0 & 0 \\
\hline
\end{tabular}
\caption{An example of $F^{(v)}$, which represents $S_v$.}\label{table2-example}
\end{table}


We introduce two subregions of $F$.
The area in which $b\leq a+1$ is referred to as \emph{Area 1}.
The area in which $b\geq a+1$ is referred to as \emph{Area 2}. 

\begin{lemma}\label{table-F-property}
\begin{enumerate}
  \item For $(a,b),(a+1,b)$ in Area 1, $F(a,b)=0$ implies $F(a+1,b)=0$.

  \item For $(a,b),(a,b+1)$ in Area 2, $F(a,b)=0$ implies $F(a,b+1)=0$.

  \item If $F(a,a+1)=0$, then $F(a+1,a+2)=0$.

  \item If $F(a+1,a+1)=0$ and $F(a+1,a+2)=1$, then $F(a,a+1)=1$.
\end{enumerate}
\end{lemma}

\begin{proof}
1. By $F(a,b)=0$, we know $cost(a,b,s)>L$ for all $s\leq b$. 
    Applying lemma~\ref{cost-lemma}.3, $cost(a+1,b,s)\geq cost(a,b,s)>L$. So $F(a+1,b)=0$.  

\smallskip \noindent 
2. By $F(a,b)=0$, we know $cost(a,b,s)>L$ for all $s\leq a+1$. 
    Applying lemma~\ref{cost-lemma}.2, $cost(a,b+1,s)\geq cost(a,b,s)>L$. So $F(a,b+1)=0$.   

\smallskip \noindent 
3. Because $F(a,a+1)=0$, we get $cost(a,a+1,s)>L$ for $s\in[0,a+1]$.
    In particular, $cost(a,a+1,a+1)>L$.
    Hence $cost(a+1,a+2,a+2)>L$ due to lemma \ref{cost-lemma}.4.
    Applying Lemmas \ref{cost-lemma}.2 and \ref{cost-lemma}.3,
    $cost(a+1,a+2,s)\geq cost(a,a+1,s)>L$ for all $s\in[0,a+1]$.
    Together, $cost(a+1,a+2,s)>L$ for all $s\in[0,a+2]$, which implies $F(a+1,a+2)=0$.

\smallskip \noindent 
4.  By $F(a+1,a+2)=1$, we know
      $\mathop{\min}_{0\leq s\leq a+2}cost(a+1,a+2,s)\leq L$.
    
    Because $F(a+1,a+1)=0$, we have $cost(a+1,a+1,s)>L$ for $s\in[0,a+1]$.
    Further with Lemma~\ref{cost-lemma}.2, we get $cost(a+1,a+2,s)>L$ for $s\in[0,a+1].$
    
    Together, we conclude that $cost(a+1,a+2,a+2)\leq L.$
    
    By Lemma~\ref{cost-lemma}.4, $cost(a,a+1,a+1)\leq cost(a+1,a+2,a+2)\leq L$. Therefore
    $\mathop{\min}_{0\leq s\leq a+1}cost(a,a+1,s)\leq L$, which further implies that $F(a,a+1)=1$.
    \qed
\end{proof}

\begin{proof}[Theorem~\ref{thm:S-more}] 
1. Suppose row $a$ are all 0's. By Lemma \ref{table-F-property}.1 and \ref{table-F-property}.3, we know $F(a+1,i)=0$ for all $i\leq a+2$. See Table~\ref{table2-1}.

\begin{table}
  \centering
  \renewcommand{\arraystretch}{1.5}
  \setlength{\tabcolsep}{5mm}
  \begin{tabular}{|c|c|c|c|c|c|c|c|c|c|}
\hline
\diagbox[width=0.14\textwidth,height=8.5mm]{$a$}{$b$} & \makebox[2mm][c]{$0$} & \makebox[2mm][c]{$1$} & $...$ & \makebox[2mm][c]{$a$} & \makebox[2mm][c]{$a+1$} & \makebox[2mm][c]{$a+2$} & $...$ & $...$ & \makebox[2mm][c]{$n_b$} \\
\hline
$0$ &  &  &  &  &  &  &  &  & \\
\hline
$1$ &  &  &  &  &  &  &  &  & \\
\hline
$2$ &  &  &  &  &  &  &  &  &  \\
\hline
$...$ &  &  &  &  &  &  &  &  &  \\
\hline
$a$ & 0 & 0 & 0 & 0 & 0 & 0 & 0 & 0 & 0 \\
\hline
$a+1$ & 0 & 0 & 0 & 0 & 0 & 0 &  &  &  \\
\hline
$...$ &  &  &  &  &  &  &  &  &  \\
\hline
$n_a$ &  &  &  &  &  &  &  &  &  \\
\hline
\end{tabular}
\caption{}\label{table2-1}
\end{table}
    
Further according to Lemma \ref{table-F-property}.2, row $a+1$ are all 0's.

\medskip 2.  If there exists $b\leq a+1$, such that $F(a+1,b)=1$, we know $F(a,b)$ and $F(a+1,b)$ are both in \emph{Area 1}. So by Lemma~\ref{table-F-property}.1 we know $F(a,b)=1$.
  
Otherwise, $F(a+1,b)=0$ for all $b\leq a+1$. So there exists $b_0> a+1$ such that $F(a+1,b_0)=1$. 
    By lemma \ref{table-F-property}.2 we know $F(a+1,a+2)=1$. So by lemma \ref{table-F-property}.4 we know $F(a,a+1)=1$.
  \qed
\end{proof}

\section{Multiple machines case}\label{sect:multiple-day}

Below we give an $O(n^2p\log(n))$ algorithm based on dynamic programming and binary search, and subsection \ref{proof of theorem3} proves its running time.

\subsection{Algorithm}\label{algorithm}
For convenience, let $b_v(a)=\left\{b\geq 0 \mid f(a,b)\leq L\right\}$. Given finite set $A,B\subseteq \mathbb{Z}$, we define $A+B=\{a+b \mid a\in A,b\in B\}$ as their \emph{sumset}.
\begin{definition}
For $v\geq 0$, $a\geq 0$ and $L\geq 0$, define $dp(v,a)$ to be the set of amount $b$ of B-jobs, 
    such that the first $v$ machines can finish the task-combination $(a,b)$ within $L$ time units. Formally,
$$dp(v,a)=\bigcup_{\begin{subarray}{c}x_1+\ldots+x_v=a \\ x_1\geq 0,\ldots,x_v\geq 0\end{subarray}}(b_1(x_1)+\ldots+b_v(x_v)).$$
\end{definition}

 We can immediately get recurrence formula for $dp(v,a)$ by enumerating the number of A-jobs finished by the last machine:
\begin{equation}\label{dp}
  dp(v,a)=\bigcup_{0\leq a'\leq a}\left(dp(v-1,a-a')+b_v(a')\right).
\end{equation}

We obtain the answer by binary searching the smallest time units $L$ that makes $dp(p,n_a)$ containing $n_b$. The following theorem assures that $dp(p,n_a)$ can be calculated in $O(n^2p)$ time. The upper bound of $L$ is $O(n^2)$ by scheduling all jobs to one machine. Therefore the answer can be calculated in $O(n^2p\log(n))$ time.
\begin{theorem}\label{thm:dp}
  Set $dp(v,a)$ is an interval.
\end{theorem}

Suppose Theorem \ref{thm:dp} holds. Let $dp(v,a)=[dp(v,a).l,dp(v,a).r]$. Theorem \ref{thm:S} states that $b_v(a)$ is an interval. Let $b_v(a)=[b_v(a).l,b_v(a).r]$. And (\ref{dp}) can be simplified as:
$$
    dp(v,a).l = min_{0\leq a'\leq a}\left(dp(v-1,a-a').l+b_v(a').l\right),
$$
$$
    dp(v,a).r = max_{0\leq a'\leq a}\left(dp(v-1,a-a').r+b_v(a').r\right).
$$

\subsection{Proof of theorem 3}\label{proof of theorem3}
To prove this theorem, we need some knowledge about $dp(v,a)$ as shown below.

\begin{lemma}\label{lemma:dp}
For any $v\geq 1$ and $a\geq 0$,
\begin{enumerate}
  \item If $dp(v,a)=\emptyset$, then $dp(v,a+\delta)=\emptyset$ for all $\delta>0$.
  \item If $dp(v,a)\neq \emptyset$ and $dp(v,a+1)\neq \emptyset$, at leat one of the following holds:\\
      (1) $dp(v,a)\cap dp(v,a+1)\neq \emptyset$.\\
      (2) there exists $b_0\in dp(v,a)$ and $b_0+1\in dp(v,a+1)$.
\end{enumerate}

\end{lemma}
\begin{proof}
1. We prove it by induction on $v$. When $v=1$, we have $dp(1,a)=b_1(a)$. Notice that Theorem \ref{thm:S-more}.1 states that if $b_v(a)=\emptyset$, then $b_v(a+\delta)=\emptyset$. So the conclusion holds.

Suppose conclusion holds for $v-1$. Recall that
  $$  dp(v,a)=\bigcup_{0\leq a'\leq a}\left(dp(v-1,a-a')+b_v(a')\right),$$
  and
  $$  dp(v,a+\delta)=\bigcup_{0\leq a'\leq a+\delta}\left(dp(v-1,a+\delta-a')+b_v(a')\right).$$
  Because $dp(v,a)=\emptyset$, we know either $dp(v-1,a-a')=\emptyset$ or $b_v(a')=\emptyset$ for $a'\in [0,a]$.
  
  In order to prove $dp(v,a+\delta)=\emptyset$, we only need to prove that either $dp(v-1,a+\delta-a')=\emptyset$ or $b_v(a')=\emptyset$ for $a'\in [0,a+\delta]$.
  \setcounter{case}{0}
  \begin{case}[$0\leq a'\leq a$]
   If $b_v(a')=\emptyset$, conclusion holds obviously. Otherwise, $dp(v-1,a-a')=\emptyset$. Because $a+\delta-a'>a-a'$, by conclusion hypothesis we know $dp(v-1,a+\delta-a')=\emptyset$.
  \end{case}
  \begin{case}[$a+1\leq a'\leq a+\delta$]
  If there exists $a''\in [0,a]$, such that $b_v(a'')=\emptyset$, we know $b_v(a')=\emptyset$ by theorem \ref{thm:S-more}.1. Otherwise, $dp(v-1,a-a'')=\emptyset$ for all $a''\in [0,a]$. Let $a''=a$, we have $dp(v-1,0)=\emptyset$. Since $a+\delta-a'>0$, by conclusion hypothesis we know $dp(v-1,a+\delta-a')=\emptyset$.
  \end{case}

2. We prove it by induction on $v$. When $v=1$, we have $dp(1,a)=b_1(a)$. Notice that Theorem \ref{thm:S-more}.2 states that, if $b_v(a)\ne \emptyset$ and $b_v(a+1)\ne \emptyset$, either $b_v(a)\cap b_v(a+1)\ne \emptyset$ or $a+1\in b_v(a)$, $a+2\in b_v(a+1)$. So the conclusion holds.

Suppose conclusion holds for $v-1$. Lemma \ref{lemma:dp}.1 states that non-empty $dp(v-1,a)$ is consecutive when $a$ changes. Theorem \ref{thm:S-more}.1 states that non-empty $b_v(a)$ is consecutive when $a$ changes. So (\ref{dp}) can be written as
\begin{equation}\label{dp-a}
  dp(v,a)=\bigcup_{a'\in [l_1,r_1]\cap [l_2,r_2]}(dp(v-1,a-a')+b_v(a')),
\end{equation}
where $dp(v-1,a-a')\ne \emptyset$ for all $a'\in [l_1,r_1]$, and $b_v(a')\ne \emptyset$ for all $a'\in [l_2,r2]$. So
\begin{equation}\label{dp-a+1}
  dp(v,a+1)=\bigcup_{a'\in [l_1+1,r_1+1]\cap [l_2,r_2]}(dp(v-1,a+1-a')+b_v(a')).
\end{equation}

\setcounter{case}{0}
\begin{case}[$l_1=r_1$]
In this case 
$$dp(v,a)=dp(v-1,a-l_1)+b_v(l_1),$$
and
$$dp(v,a+1)=dp(v-1,a-l_1)+b_v(l_1+1).$$

If $b_v(l_1)\cap b_v(l_1+1)\neq \emptyset$, we know $dp(v,a)\cap dp(v,a+1)\neq \emptyset$. Otherwise, according to theorem \ref{thm:S-more}.2, we have $l_1+1\in b_v(l_1)$ and $l_1+2\in b_v(l_1+1)$. Suppose $b_0\in dp(v-1,a-l_1)$, then $b_0+l_1+1\in dp(v,a)$ and $b_0+l_1+2\in dp(v,a+1)$.
\end{case}
\begin{case}[$l_1<r_1$]
Notice that in this case, the intersection of $[l_1,r_1]\cap [l_2,r_2]$ and $[l_1+1,r_1+1]\cap [l_2,r_2]$ is non-empty. So there exists $a'$, such that $dp(v-1,a-a')+b_v(a')$ and $dp(v-1,a+1-a')+b_v(a')$ are both non-empty.

If $dp(v-1,a-a')\cap dp(v-1,a+1-a')\neq \emptyset$, then $dp(v,a)\cap dp(v,a+1)\neq \emptyset$. Otherwise, according to induction hypothesis, there exists $b_0\in dp(v-1,a-a')$ and $b_0+1\in dp(v-1,a+1-a')$. Suppose $b_1\in b_v(a')$, we know $b_0+b_1\in dp(v,a)$ and $b_0+b_1+1\in dp(v,a+1)$.  \qed
\end{case}
\end{proof}

We can now prove Theorem \ref{thm:dp}.
\begin{proof}[of Theorem \ref{thm:dp}]
  We prove it by induction on $v$. When $v=1$, we have $dp(1,a)=b_1(a)$. The conclusion holds by Theorem \ref{thm:S}. 
  
  Suppose conclusion holds for $v-1$. Lemma \ref{lemma:dp}.1 states that non-empty $dp(v,a)$ is consecutive when $a$ changes. Theorem \ref{thm:S-more}.1 states that non-empty $b_v(a)$ is consecutive when $a$ changes. So (\ref{dp}) can be written as
   $$dp(v,a)=\bigcup_{a'\in [l,r]}(dp(v-1,a-a')+b_v(a')),$$
   where $dp(v-1,a-a')+b_v(a')\neq \emptyset$ for all $a'\in [l,r]$.
   
   To prove $dp(v,a)$ an interval, it reduces to prove that for all $a'\in[l,r-1]$, $dp(v-1,a-a')+b_v(a')\cup dp(v-1,a-(a'+1))+b_v(a'+1)$ is an interval. According to theorem \ref{thm:S-more}.2 and lemma \ref{lemma:dp}.2, there are four cases we need to discuss.
   \setcounter{case}{0}
   \begin{case} [$b_v(a')\cap b_v(a'+1)\neq \emptyset;$\\
\indent\indent$dp(v-1,a-a')\cap dp(v-1,a-(a'+1))\neq \emptyset$]
   
   The intersection of $dp(v-1,a-a')+b_v(a')$ and $dp(v-1,a-(a'+1))+b_v(a'+1)$ is non-empty. By conclusion hypothesis and theorem \ref{thm:S}, their union is an interval.
   \end{case}
   
   \begin{case} [$a'+1\in b_v(a'),a'+2\in b_v(a'+1);$\\
\indent\indent$dp(v-1,a-a')\cap dp(v-1,a-(a'+1))\neq \emptyset$]

   Suppose $b_0\in dp(v-1,a-a')\cap dp(v-1,a-(a'+1))$, we have $b_0+a'+1\in dp(v-1,a-a')+b_v(a')$, and $b_0+a'+2\in dp(v-1,a-(a'+1))+b_v(a'+1)$. By conclusion hypothesis and theorem \ref{thm:S}, their union is an interval.
   \end{case}
   
   \begin{case} [$b_v(a')\cap b_v(a'+1)\neq \emptyset;$\\
\indent\indent$\mbox{ there exists }b_0\in dp(v-1,a-(a'+1)),b_0+1\in dp(v-1,a-a')$]
   
   Proof is same as in case 2.
   
   \end{case}
   
   \begin{case} [$a'+1\in b_v(a'),a'+2\in b_v(a'+1);$\\
\indent\indent$\mbox{ there exists }b_0\in dp(v-1,a-(a'+1)),b_0+1\in dp(v-1,a-a')$]

   Then $b_0+a'+2\in dp(v-1,a-a')+b_v(a')$, and $b_0+a'+2\in dp(v-1,a-(a'+1))+b_v(a'+1)$. So their intersection is non-empty. By conclusion hypothesis and theorem \ref{thm:S}, their union is an interval. \qed
   
   \end{case}
\end{proof}

\subsection{The linear case}
It's easy to know that there exists optimal solution that each machine processes several A-jobs and then several B-jobs. Therefore
\begin{equation*}
     f(a,b)=t^A_v[a>0]+t^B_v[b>0]+k^A_va+k^B_vb,
\end{equation*}
where $[\cdot]$ denotes the Iverson bracket.
And
\begin{equation*}
  \begin{split}
     b_v(a)&=\{b\geq 0:f(a,b)\leq L\} \\
       &=\left[0,\frac{1}{k^B_v}(L-t^A_v[a>0]-t^B_v[b>0])-\frac{k^A_v}{k^B_v}a\right]. \\
  \end{split}
\end{equation*}
Set $b_v(a)$ is obviously an interval and satisfies Theorem \ref{thm:S-more}. And $dp(v,a)$ satisfies Theorem \ref{thm:dp}. 

\bibliographystyle{unsrt}
\bibliography{TASKSCH}

\begin{thebibliography}{10}

\bibitem{Lawler}
Eugene~L Lawler, Jan~Karel Lenstra, Alexander HG~Rinnooy Kan, and David~B
  Shmoys.
\newblock Sequencing and scheduling: Algorithms and complexity.
\newblock {\em Handbooks in operations research and management science},
  4:445--522, 1993.

\bibitem{Garey}
Michael~R Garey and David~S Johnson.
\newblock ``strong''np-completeness results: Motivation, examples, and
  implications.
\newblock {\em Journal of the ACM (JACM)}, 25(3):499--508, 1978.

\bibitem{Lenstra}
Jan~Karel Lenstra, David~B Shmoys, and {\'E}va Tardos.
\newblock Approximation algorithms for scheduling unrelated parallel machines.
\newblock {\em Mathematical programming}, 46:259--271, 1990.

\bibitem{Ghirardi}
Marco Ghirardi and Chris~N Potts.
\newblock Makespan minimization for scheduling unrelated parallel machines: A
  recovering beam search approach.
\newblock {\em European Journal of Operational Research}, 165(2):457--467,
  2005.

\bibitem{Alon}
Noga Alon, Yossi Azar, Gerhard~J Woeginger, and Tal Yadid.
\newblock Approximation schemes for scheduling.
\newblock In {\em SODA}, pages 493--500. Citeseer, 1997.

\bibitem{Azar}
Yossi Azar and Amir Epstein.
\newblock Convex programming for scheduling unrelated parallel machines.
\newblock In {\em Proceedings of the thirty-seventh annual ACM symposium on
  Theory of computing}, pages 331--337, 2005.

\bibitem{Kumar}
VS~Anil Kumar and Madhav~V Marathe.
\newblock Approximation algorithms for scheduling on multiple machines.
\newblock In {\em 46th Annual IEEE Symposium on Foundations of Computer Science
  (FOCS'05)}, pages 254--263. IEEE, 2005.

\bibitem{Im}
Sungjin Im and Shi Li.
\newblock Better unrelated machine scheduling for weighted completion time via
  random offsets from non-uniform distributions.
\newblock In {\em 2016 IEEE 57th Annual Symposium on Foundations of Computer
  Science (FOCS)}, pages 138--147. IEEE, 2016.

\bibitem{Bhaskara}
Aditya Bhaskara, Ravishankar Krishnaswamy, Kunal Talwar, and Udi Wieder.
\newblock Minimum makespan scheduling with low rank processing times.
\newblock In {\em Proceedings of the twenty-fourth annual ACM-SIAM symposium on
  Discrete algorithms}, pages 937--947. SIAM, 2013.

\bibitem{Chen}
Lin Chen, D{\'a}niel Marx, Deshi Ye, and Guochuan Zhang.
\newblock Parameterized and approximation results for scheduling with a low
  rank processing time matrix.
\newblock In {\em 34th Symposium on Theoretical Aspects of Computer Science
  (STACS 2017)}. Schloss-Dagstuhl-Leibniz Zentrum f{\"u}r Informatik, 2017.

\bibitem{Deng}
Shichuan Deng, Jian Li, and Yuval Rabani.
\newblock Generalized unrelated machine scheduling problem.
\newblock In {\em Proceedings of the 2023 Annual ACM-SIAM Symposium on Discrete
  Algorithms (SODA)}, pages 2898--2916. SIAM, 2023.

\bibitem{Bamas}
{\'E}tienne Bamas, Alexander Lindermayr, Nicole Megow, Lars Rohwedder, and Jens
  Schl{\"o}ter.
\newblock Santa claus meets makespan and matroids: Algorithms and reductions.
\newblock In {\em Proceedings of the 2024 Annual ACM-SIAM Symposium on Discrete
  Algorithms (SODA)}, pages 2829--2860. SIAM, 2024.

\end{thebibliography}

\section*{Acknowledgement}
We sincerely appreciate Ruixi Luo for his proof reading and modification advise.

\clearpage
\appendix

\section{A proof of Lemma~\ref{lemma:cost^A-diff}}
\begin{proof}[of Lemma~\ref{lemma:cost^A-diff}]
  First we prove $cost^A(a,s)$ is convex function of $s$. The proof of $cost^B(b,s)$ is the same.
  
  Let $a=sq+r$, where $0\leq r<s$, we have 
  $$
  cost^A(a,s)=\left[r(q+1)^2+(s-r)q^2\right]k^A_v+st^A_v.
  $$
  Let $g(s)=r(a+1)^2+(s-r)q^2$, it reduces to prove that $g(s)$ is convex.
  
  Let
  \begin{equation}\label{division-s}
    a=sq+r\mbox{, where }0\leq r<s,
  \end{equation}
  \begin{equation}\label{division-s+1}
    a=(s+1)q'+r'\mbox{, where }0\leq r'<s+1,
  \end{equation}
  \begin{equation}\label{division-s+2}
    a=(s+2)q''+r''\mbox{, where }0\leq r''<s+2.
  \end{equation}
    Notice that $q\geq q'\geq q''$. 
    
    Then let
  \begin{equation*}
    \begin{split}
       \Delta_s &= g(s+1)-g(s) \\
         &= 2r'q'+r'+(s+1)q'^2-2rq-r-sq^2,
    \end{split}
  \end{equation*}
  \begin{equation*}
    \begin{split}
       \Delta_{s+1} &= g(s+2)-g(s+1) \\
         &= 2r''q''+r''+(s+2)q''^2-2r'q'-r'-(s+1)q'^2.
    \end{split}
  \end{equation*}
  We only need to prove
  $$\Delta_{s+1}\geq \Delta_{s}.$$
  
  \setcounter{case}{0}
  \begin{case}[$q=q'=q''$]
  By (\ref{division-s}), (\ref{division-s+1}) and (\ref{division-s+2}) we know $r'-r=r''-r'=-q$.
  And we have 
   $$\Delta_s=(2q+1)(r'-r)+q^2,$$
   and 
   $$\Delta_{s+1}=(2q+1)(r''-r')+q^2.$$
   Therefore $\Delta_s=\Delta_{s+1}$.
  \end{case}
  
  \begin{case}[$q>q'=q''$]
    
\begin{figure}
  \centering
  \includegraphics[width=1.2\textwidth]{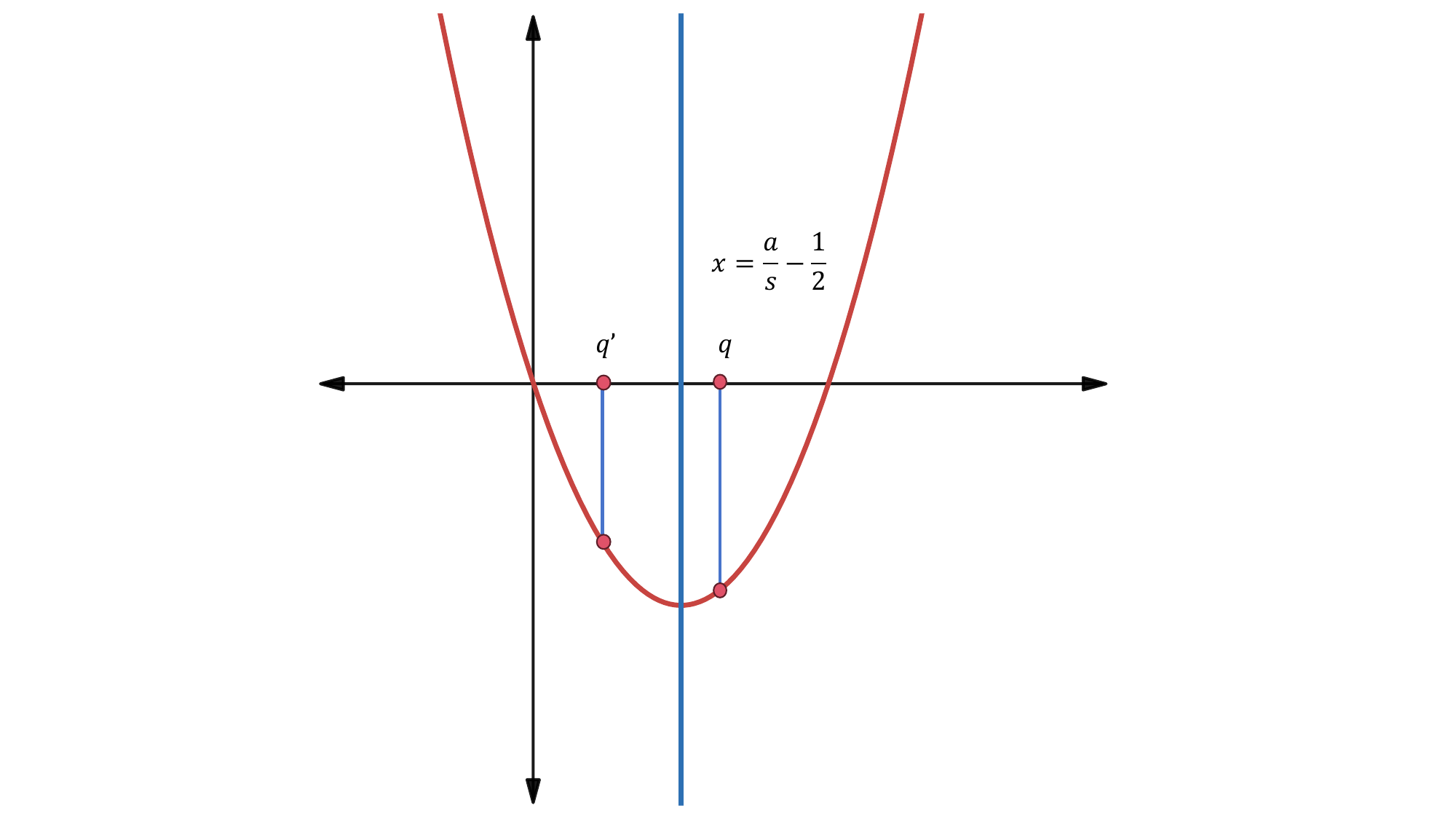}
  \caption{}\label{fig-lemma1}
\end{figure}

  We have
  $$\Delta_s=2r'q'+r'+(s+1)q'^2-2rq-r-sq^2,$$
  and
  \begin{equation*}
    \begin{split}
       \Delta_{s+1} &= 2r''q''+r''+(s+2)q''^2-2r'q'-r'-(s+1)q'^2 \\
         &= 2r''q'+r''+q'^2-2r'q'-r'.
    \end{split}
  \end{equation*}
  Apply (\ref{division-s}), (\ref{division-s+1}) and (\ref{division-s+2}) to the above formulas and eliminate $r$, $r'$ and $r''$. We have
  $$\Delta_{s+1}-\Delta_s=\left[-sq^2+(2a-s)q\right]-\left[-sq'^2+(2a-s)q'\right].$$
  Let $h(x)=sx^2-(2a-s)x$, it reduces to prove that $h(q)\leq h(q')$.
  
  Because $q>q'$, let $q=q'+\alpha$, where $\alpha\geq 1$.

  Notice that $h(x)$ is a quadratic function with axis of symmetry $x=\frac{a}{s}-\frac{1}{2}$. See figure \ref{fig-lemma1}. We consider the middle point between $q'$ and $q$.
  \begin{equation*}
    \begin{split}
       \frac{q+q'}{2} &= \frac{2q-\alpha}{2} \\
         &= q-\frac{\alpha}{2} \\
         &= \lfloor \frac{a}{s} \rfloor - \frac{\alpha}{2} \\
         &\leq \frac{a}{s} - \frac{\alpha}{2} \\
         &\leq \frac{a}{s} - \frac{1}{2}.
    \end{split}
  \end{equation*}
  Therefore $h(q)\leq h(q')$.
  \end{case}
  \begin{case}[$q=q'>q''$]
     
\begin{figure}
  \centering
  \includegraphics[width=1.2\textwidth]{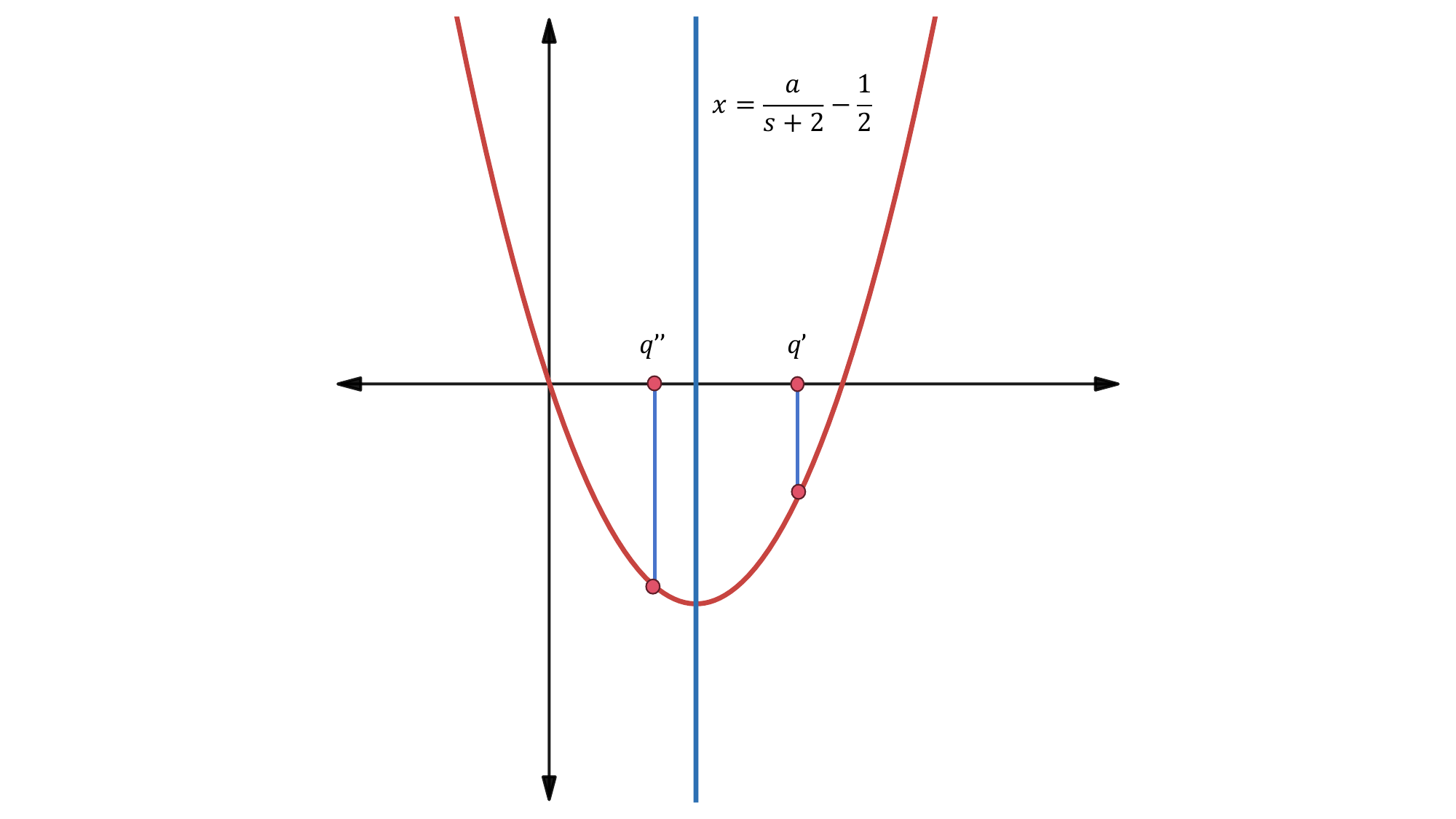}
  \caption{}\label{fig-lemma2}
\end{figure}

    We have
    \begin{equation*}
      \begin{split}
         \Delta_s &= 2r'q'+r'+(s+1)q'^2-2rq-r-sq^2 \\
           &= 2(r'-r)q'+r'-r+q'^2,
      \end{split}
    \end{equation*}
    and 
    \begin{equation*}
      \Delta_{s+1} = 2r''q''+r''+(s+2)q''^2-2r'q'-r'-(s+1)q'^2.
    \end{equation*}
  Apply (\ref{division-s}), (\ref{division-s+1}) and (\ref{division-s+2}) to the above formulas and eliminate $r$, $r'$ and $r''$. We have
  $$\Delta_{s+1}-\Delta_s=-(s+2)q''^2+(2a-s-2)q''+(s+2)q'^2-(2a-s-2)q'.$$
  Let $h(x)=(s+2)x^2-(2a-s-2)x$, it reduces to prove that $h(q'')\leq h(q')$.
  
  Because $q'>q''$, let $q'=q''+\alpha$, where $\alpha\geq 1$.

  Notice that $h(x)$ is a quadratic function with axis of symmetry $x=\frac{a}{s+2}-\frac{1}{2}$. See figure \ref{fig-lemma2}. We consider the middle point between $q'$ and $q''$. We have
  \begin{equation*}
    \begin{split}
       \frac{q'+q''}{2} &= \frac{2q''+\alpha}{2} \\
         &= q''+\frac{\alpha}{2} \\
         &= \lfloor \frac{a}{s+2} \rfloor + \frac{\alpha}{2}, \\
\mbox{let }\frac{a}{s+2}+\beta &\mbox{, where }0\leq \beta \leq 1\mbox{, then }\\
         &= \frac{a}{s+2}-\beta+\frac{\alpha}{2} \\
         &> \frac{a}{s+2}-1+\frac{1}{2} \\
         &= \frac{a}{s+2}-\frac{1}{2}.
    \end{split}
  \end{equation*}
  Therefore $h(q'')<h(q')$.
  \end{case}
  
  \begin{case}[$q>q'>q''$]
     
  Recall that
    \begin{equation*}
         \Delta_s = 2r'q'+r'+(s+1)q'^2-2rq-r-sq^2,
  \end{equation*}
  \begin{equation*}
         \Delta_{s+1} = 2r''q''+r''+(s+2)q''^2-2r'q'-r'-(s+1)q'^2.
  \end{equation*}
   Apply (\ref{division-s}), (\ref{division-s+1}) and (\ref{division-s+2}) to the above formulas and eliminate $r$, $r'$ and $r''$. We have
   \begin{equation*}
   \begin{split}
      \Delta_{s+1}-\Delta_s=(s+2)q'^2-&(2a-s-2)q'-\left[sq^2-(2a-s)q\right]- \\
      &\left[(s+2)q''^2-(2a-s-2)q''-(sq'^2-(2a-s)q')\right].
   \end{split}
   \end{equation*}
  \end{case}
  Let $h(x_1,x_2)=(s+2)x_1^2-(2a-s-2)x_1-sx_2^2+(2a-s)x_2$, it reduces to prove that $h(q'',q')\leq h(q',q)$.
  
  Notice that $h(x_1,x_2)$ is a hyperbolic paraboloid, we reduce it to its general form:
  \begin{equation*}
  \begin{split}
     h(x_1,x_2)+\frac{\frac{1}{4}(2a-s-2)^2}{s+2}&-\frac{\frac{1}{4}(2a-s)^2}{s}=\\
      &\frac{\left(x_1-\frac{\frac{1}{2}(2a-s-2)}{s+2}\right)^2}{\left(\frac{1}{\sqrt{s+2}}\right)^2}-\frac{\left(x_2-\frac{\frac{1}{2}(2a-s)}{s}\right)^2}{\left(\frac{1}{\sqrt{s}}\right)^2}.
  \end{split}
  \end{equation*}
  We can observe its saddle point $(x_1^*,x_2^*,h^*)$, where
  \begin{equation*}
    \begin{split}
       x_1^* &= \frac{\frac{1}{2}(2a-s-2)}{s+2}, \\
       x_2^* &= \frac{\frac{1}{2}(2a-s)}{s}, \\
       h^* &= \frac{\frac{1}{4}(2a-s)^2}{s}-\frac{\frac{1}{4}(2a-s-2)^2}{s+2}.
    \end{split}
  \end{equation*}
     
\begin{figure}
  \centering
  \includegraphics[width=1.2\textwidth]{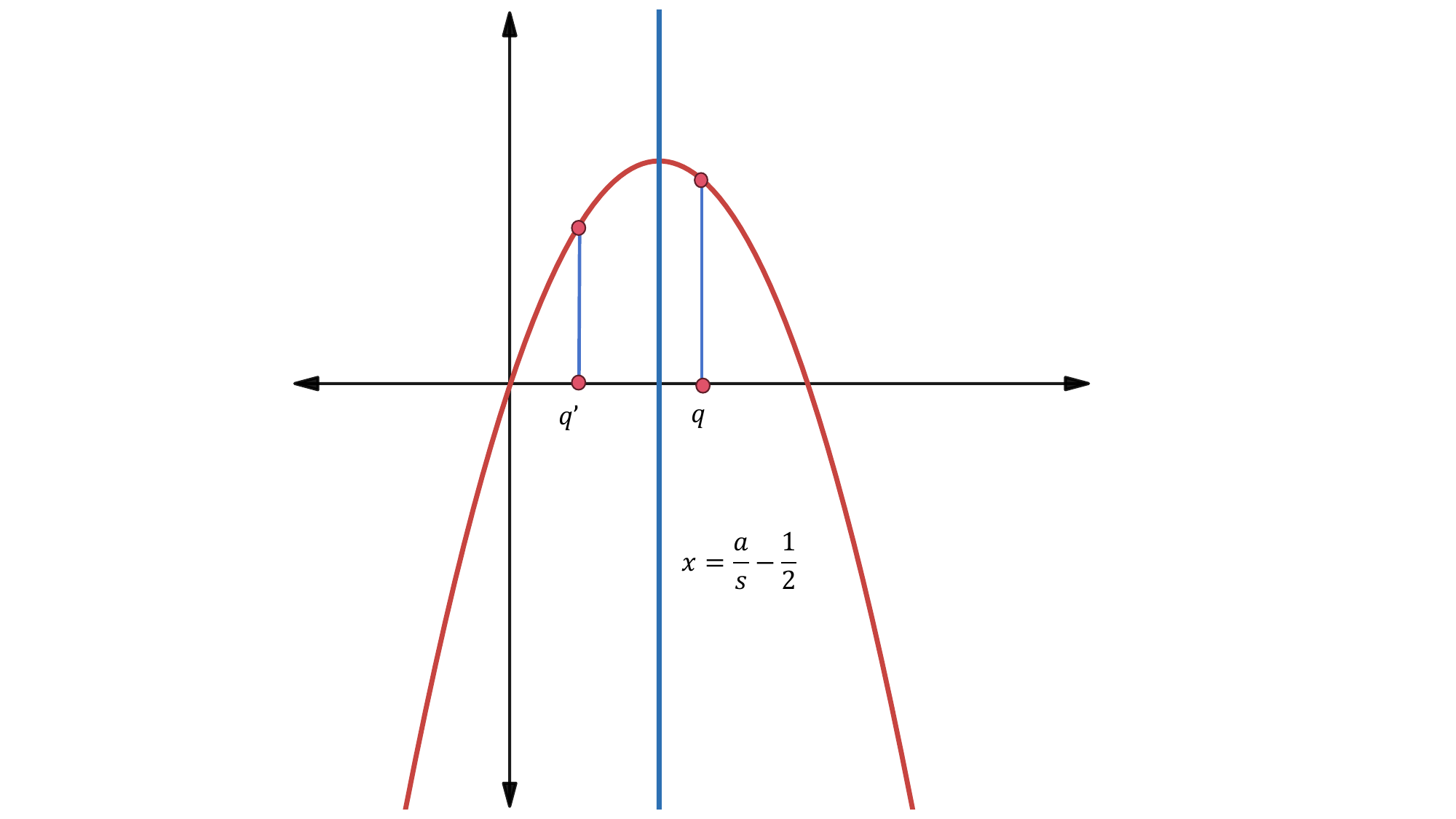}
  \caption{}\label{fig-lemma4}
\end{figure}
  
  First, consider the cross section of the hyperbolic paraboloid parallel to plane $x_1Oh$: $h(x_1,*)$. Notice that it's a parabola with axis of symmetry $x=\frac{a}{s+2}-\frac{1}{2}$. Consider $q'$ and $q''$. Same as case 3, we know $h(q',*)\geq h(q'',*)$.
  
  Next, consider the cross section of the hyperbolic paraboloid parallel to plane $x_2Oh$: $h(*,x_2)$. Notice that it's a parabola with axis of symmetry $x=\frac{a}{s}-\frac{1}{2}$. See figure \ref{fig-lemma4}. Consider the middle point of $q$ and $q'$.
  
  Let $q=q'+\alpha$, where $\alpha\geq 1$. We have
  \begin{equation*}
    \begin{split}
       \frac{q'+q}{2} &= \frac{2q-\alpha}{2} \\
         &= q-\frac{\alpha}{2} \\
         &= \lfloor\frac{a}{s}\rfloor-\frac{\alpha}{2} \\
         &\leq \frac{a}{s}-\frac{1}{2}.
    \end{split}
  \end{equation*}
  Therefore $h(*,q)\geq h(*,q')$.
  
  By $h(q',*)\geq h(q'',*)$ and $h(*,q)\geq h(*,q')$, we know
  $$h(q',q)\geq h(q'',q)\geq h(q'',q').$$
  
  At last we prove that $mcost^A(a,s)$ is convex function of $s$. Recall that
  $$mcost^A(a,s)=\min_{s'\in \{s-1,s,s+1\},s'\geq 0}cost^A(a,s').$$
  \setcounter{case}{0}
  \begin{case}[$cost^A(a,s)$ is increasing (or decreasing) for $s$]

  Then $mcost^A(a,s)=cost^A(a,s-1)$ (or $cost^A(a,s+1)$).
  \end{case}
  \begin{case}[There exists $s_0$, such that when $s\leq s_0$, $cost^A(a,s)$ decreases, and when $s\geq s_0$, $cost^A(a,s)$ increases]
  
  According to case 1, $mcost^A(a,s)$ is convex when $s\leq s_0$ and $s\geq s_0$. It reduces to prove that 
  \begin{equation}\label{mcostA-diff-proof}
    mcost^A(a,s_0)-mcost^A(a,s_0-1)\leq mcost^A(a,s_0+1)-mcost^A(a,s_0).
  \end{equation}
  Notice that
  $$mcost^A(a,s_0-1)=mcost^A(a,s_0)=mcost^A(a,s_0+1)=cost^A(a,s_0),$$
  so the left and right side of (\ref{mcostA-diff-proof}) are both $0$.
  
  Therefore $mcost^A(a,s)$ is convex.  \qed
  \end{case}
\end{proof}

\end{document}